\documentclass[11pt]{amsart}
\usepackage{graphicx}
\usepackage{amssymb}
\usepackage{amsthm}
\usepackage{natbib}
\usepackage{booktabs}
\usepackage[margin=1.2in]{geometry}
\usepackage{enumerate}
\usepackage{setspace}
\newtheorem{theorem}{Theorem}
\newtheorem{lemma}{Lemma}

\newtheorem{proposition}{Proposition}

\newtheorem{corollary}{Corollary}

\newcommand{\abs}[1]{\ensuremath{\left|#1\right|}}

\newcommand{\bY}{\mathbf{Y}}
\newcommand{\bX}{\mathbf{X}}
\newcommand{\bbeta}{\boldsymbol{\beta}}

\DeclareMathOperator\sign{sign}

\title[]{Bayesian variable selection and estimation based on global-local shrinkage priors}
\author{Xueying Tang, Xiaofan Xu, Malay Ghosh and Prasenjit Ghosh}
\date{}                                           

\begin{document}

\begin{abstract}
In this paper, we consider Bayesian variable selection problem of linear regression model with global-local shrinkage priors on the regression coefficients. We propose a variable selection procedure that select a variable if the ratio of the posterior mean to the ordinary least square estimate of the corresponding coefficient is greater than $1/2$. Under the assumption of orthogonal designs, we show that if the local parameters have polynomial-tailed priors, our proposed method enjoys the oracle property in the sense that it can achieve variable selection consistency and optimal estimation rate at the same time. However, if, instead, an exponential-tailed prior is used for the local parameters, the proposed method does not have the oracle property.
\end{abstract}
\maketitle

\section{Introduction}
The objective of this article is simultaneous variable selection and estimation in linear regression models under global-local shrinkage priors and a suitable thresholding. Selection of the best available model among a set of candidate models is extremely useful for most statistical applications. The problem often reduces to the choice of a subset of variables from all predictive variables in a regression setting. Linear regression models continue to occupy a prominent place in variable selection problems due to their interpretability as well as analytical tractability. Throughout this paper, we consider classical linear regression models with response vector $\mathbf{Y} = (y_1, \ldots, y_n)$ and a set of predictors $\mathbf{x}_1, \ldots, \mathbf{x}_p$. The target is to fit a model of the form
\begin{equation}
\mathbf{Y} = \sum_{i=1}^{p} \mathbf{x}_i\beta_i + \boldsymbol{\epsilon} = \mathbf{X}\boldsymbol{\beta} +\boldsymbol{\epsilon},
\end{equation}
where $\mathbf{X} = (\mathbf{x}_1, \ldots, \mathbf{x}_p)$, $\boldsymbol{\beta} = (\beta_1, \ldots, \beta_p)^T$, and $\boldsymbol\epsilon \sim N(\mathbf{0}_n, \sigma^2\mathbf{I}_n)$. The goal of variable selection is to pick only a subset of predictors that are relevant for predicting a given response. 

Historically, penalized regression methods have been very successful for variable selection. In its general form, the problem reduces to minimization of the objective function
\begin{equation}\label{eq:obj_fun}
f(\boldsymbol{\beta}) = \| \mathbf{Y} - \mathbf{X}\boldsymbol{\beta} \|^2 + \lambda \sum_{i=1}^p u(\beta_i),
\end{equation}
where $\lambda$ is the penalty parameter. The choice $u(z) = z^2$ leads to the ridge \citep{marquardt1975ridge} estimator, while $u(z) = | z |$ leads to the lasso \citep{tibshirani1996regression} estimator of $\boldsymbol{\beta}$. One of the advantages of the lasso estimator is that it can produce exact zero estimates for some of the regression coefficients. Despite this distinctive feature, the lasso method has some limitations in its original form. \citet{zou2006adaptive} showed that lasso estimators could not achieve consistent variable selection and optimal estimation rate at the same time. He proposed instead the adaptive lasso which heavily penalized zero coefficients and moderately penalized large coefficients using data dependent weights for different coefficients. Specifically, the adaptive lasso estimates are found as
\begin{equation}
\hat{\boldsymbol{\beta}}^{\textrm{adap}} = \underset{\boldsymbol{\beta}}{\operatorname{argmin}}  \left\{ \|\mathbf{Y} - \mathbf{X}\boldsymbol{\beta} \|^2 + \lambda \sum_{i=1}^p | \beta_i | / | \hat{\beta}_i |^{\gamma} \right\}
\end{equation}
for some $\gamma > 0$ where $\hat{\beta}_i$ is the least squares estimator of $\beta_i$. The adaptive lasso enjoys the oracle property in the sense that it achieves simultaneously variable selection consistency and asymptotic normality with $\sqrt{n}$ convergence rate. One important feature, implicit in the proof of Theorem 2 of \cite{zou2006adaptive}, is that $\hat{\beta}_i^{\mathrm{adap}} / \hat{\beta}_i \overset{P}{\rightarrow} 0$ or $1$ according as the true coefficient value equals or is different from zero. This property is the main motivation for us to develop a new thresholding method in a Bayesian context.

The history of Bayesian variable selection goes a long way back starting with \citet{mitchell1988bayesian}, where a spike and slab prior is used for the coefficients. The spike part of their prior placed probability mass at zero to exclude irrelevant variables, while the slab part used a uniform distribution with a large symmetric range to include the important variables. Since then, many priors invented for variable selection possess the spike-and-slab feature, although they differ in the choice of the distributions for the two parts. \citet{george1993variable} proposed stochastic search variable selection (SSVS) which assumed $\beta_i$ to be a mixture of two normal distributions with different variances. The spike part is the one with a smaller variance while the slab part is the one with a much larger variance. Different from the above, \citet{geweke1996variable} used positive mass at 0 for the spike part and a normal distribution for the slab part. Another example is \citet{narisetty2014bayesian}. \citet{xu2015bayesian} also considered spike-and-slab priors, but used median thresholding to select variables in the group lasso framework.


An alternative approach to Bayesian variable selection is to use shrinkage priors for the regression coefficients. An early example of this approach is the Bayesian lasso introduced by \citet{park2008bayesian}. They performed a full Bayesian analysis analogous to lasso, interpreting $\| \mathbf{Y} - \mathbf{X}\boldsymbol{\beta} \|^2$ in \eqref{eq:obj_fun} as the negative of a multiple of the log-likelihood and the penalty function $\lambda \sum_{i=1}^p |\beta_i|$ as the negative of a double exponential prior for $\beta_i$. Following a similar idea, \citet{li2010bayesian} proposed Bayesian elastic net.

Unlike the spike-and-slab priors, shrinkage priors cannot naturally produce exact zero estimates of the regression coefficients with positive probability. Thus a critical question to answer when using shrinkage priors for variable selection is how to actually select relevant variables. \citet{li2010bayesian} presented the credible interval criterion which selects predictor $\mathbf{x}_i$ if the credible interval of $\beta_i$ does not cover 0. A criterion called scaled neighborhood criterion is also considered in \citet{li2010bayesian}. It selects predictors with posterior probability of belonging to $[-\sqrt{\mathrm{Var}(\beta_j \, |\, \mathbf{Y})}, \sqrt{\mathrm{Var}(\beta_j \, | \, \mathbf{Y})}]$ less than a certain threshold. These authors did not address the issue of any oracle property of their procedures.
\citet{bondell2012consistent} employed the conjugate normal priors for $\beta_i$ and used sparse solutions within posterior credible regions to perform selection. They gave a theoretical proof of variable selection consistency of their method. Recently, \citet{hahn2015decoupling} proposed selection of variables by minimizing the ``decoupled shrinkage and selection'' loss function after finding the posterior mean of $\bbeta$. However, a surrogate optimization problem has to be used since the original one is intractable in the presence of moderate to large number of predictors.

In this paper, we consider the problem of variable selection and estimation using global-local shrinkage priors. The prior of \citet{park2008bayesian} came as special cases. Specifically, we assume the prior distribution of $\beta_i$ is a scale mixture of normals:
$$\beta_i \, | \, \overset{ind}{\sim} \mbox{N}(0, \sigma^2\gamma_i\tau), ~ \gamma_i \overset{ind}{\sim} \pi(\gamma_i).$$ These priors approximate the spike-and-slab priors, but instead are symmetric, unimodal and absolutely continuous. They place significant probability mass around zero and have heavy tails to signify the inclusion of relevant variables. The local parameters $\gamma_i$ control the degree of shrinkage of each individual $\beta_i$ while the global parameter causes an overall shrinkage. We give a list of such priors in a later section. The list includes not only the now famous horseshoe prior of \citet{carvalho2010horseshoe}, but several other priors considered, for example, in \citet{griffin2010inference,griffin2011bayesian}, \citet{polson2010shrink,polson2012half} and \citet{armagan2011generalized,ADL2012}. We also find it convenient to classify the priors $\pi(\gamma_i)$ into two subclasses: those with exponential tails and those with polynomial tails. We propose a thresholding procedure to select relevant variables in the model. It turns out that the theoretical properties of our proposed method are closely related to the tails of $\pi(\gamma_i)$. As we will show in the subsequent sections, if polynomial-tailed priors are used, the proposed method attains the oracle property for certain choice of $\tau$ in the same sense as the adaptive lasso. In contrast, the exponential-tailed priors, while attaining variable consistency for some choice of $\tau$, will fail to attain asymptotic normality at the $\sqrt{n}$ rate.

The outline of the remaining sections is as follows. The general class of shrinkage priors and the thresholding procedure are described in Section 2. In Section 3, we present the theoretical properties of the proposed method for orthogonal designs. Section 4 contains some simulation results. Some final remarks are made in Section 5. The technical proofs of the results are deferred to the Appendix.

We want to highlight some of the findings of our paper and also compare and contrast the same with some of the other variable selection procedures. The thresholding approach used in our paper attains simultaneous variable selection and estimation with exact zero estimates of some of the regression coefficients similar to the original lasso. While variable selection consistency has been addressed in a large number of papers including situations where $p\gg n$, the asymptotic normality of the non-zero vector of regression coefficients, to our knowledge, has not been considered earlier in this generality. Moreover, although the exponential-tailed priors including those of \citet{park2008bayesian} have been addressed quite frequently in the literature, their asymptotic non-optimality, as pointed by us, has not been addressed before. Finally, despite the fact that the oracle properties are proved only for orthogonal designs, the proposed selection mechanism works for non-orthogonal designs as well as shown in the simulations of Section 4.

\section{Global local shrinkage priors and proposed Method}
For clarity, we reiterate the model considered in this article:
\begin{align} \tag{M1}\label{eq:model1}
\mathbf{Y} &\sim N(\mathbf{X}\boldsymbol\beta, \sigma^2\mathbf{I}_n),\\
\tag{M2}\label{eq:model2} \beta_i | \gamma_i & \overset{ind}{\sim} N(0, \sigma^2\gamma_i \tau),~ i= 1, \ldots, p,\\ 
\gamma_i & \overset{ind}{\sim} \pi(\gamma_i),~ i = 1, \ldots, p.\tag{M3}\label{eq:model3}
\end{align}
Throughout this article, we assume $p = p_n \leq n$. 

Many priors in Bayesian literature can be expressed in the form of scale mixture of normals, as in \eqref{eq:model2} and \eqref{eq:model3}. Table \ref{table:priors}, given later in Section 3, presents a list of such priors (of $\beta_i$) and the corresponding form of $\pi(\gamma_i)$. 
By employing two levels of parameters to express the variances in \eqref{eq:model2}, the global-local shrinkage priors assign large probabilities around zero while assigning non-trivial probabilities to values far from zero. The global parameter $\tau$ tends to shrink all $\beta_i$'s towards zero. At the same time, the local parameters $\gamma_i$ control the degree of shrinkage of each individual $\beta_i$. If $\pi(\gamma_i)$ is appropriately heavy tailed, the coefficients of important variables can be left almost unshrunk.

In the same spirit as \citet{park2008bayesian}, placing a prior on $\beta_i$ is closely related with adding a penalty term of $\beta_i$ to the ordinary least square objective function, so the properties of penalized regression estimators can shed light on the features of Bayesian estimator of $\beta_i$. The proof of Theorem 2 in \citet{zou2006adaptive} implies that under mild conditions, 
\begin{align}
\frac{\hat{\beta}_{i}^{adap}}{\hat{\beta}_{i}}\stackrel{p}{\rightarrow}\begin{cases}
0, & \text{when}\beta_{i}^{0}=0,\\
1, & \text{when }\beta_{i}^{0}\neq0,
\end{cases}
\label{eq:ratio_converge}
\end{align}
where $\beta_i^0$ is the true value of $\beta_i$ and $\hat{\beta}_{i}$ is the ordinary least square estimator of $\beta_{i}$. This indicates that the adaptive lasso estimator for the coefficient of an irrelevant variable converges to zero faster than the least square estimator. In fact, \eqref{eq:ratio_converge} holds by replacing the adaptive lasso estimator with any penalized regression estimator that has the oracle property as given in \citet{zou2006adaptive}. Many of these estimators can also be interpreted as posterior modes of priors specified in \eqref{eq:model1}-\eqref{eq:model3}. Due to the asymptotic closeness of posterior means and posterior modes under such priors, one can threshold the ratio of posterior mean and least square estimator to obtain an oracle variable selection procedure even though the posterior mean is not sparse. Motivated by this, we propose to select predictor $\mathbf{x}_{i}$ if 
\begin{equation}
\abs{\frac{\hat{\beta}_{i}^{PM}}{\hat{\beta}_{i}}}>\frac{1}{2},\label{eq:Criterion}
\end{equation}
where $\hat{\beta}_{i}^{PM}$
is the posterior mean of $\beta_{i}$ under certain shrinkage
prior. We refer to this procedure as Half Thresholding (HT) and define the HT estimator of $\beta_i$ as 
$$ \hat{\beta}_i^{HT}=\hat{\beta}_i^{PM}I\left(\abs{\frac{\hat{\beta}_{i}^{PM}}{\hat{\beta}_{i}}}>\frac{1}{2}\right). $$

Our proposed HT procedure is simple and easy to implement. Once the posterior mean and the ordinary least square estimate of $\bbeta$ are obtained, variable selection can be performed without any extra optimization step as required for example in \citet{bondell2012consistent} and \citet{hahn2015decoupling}. Besides its simplicity, as we will show in the next section, the HT procedure enjoys oracle properties for orthogonal designs if the global parameter $\tau$ and the prior of $\gamma_i$ are chosen appropriately.

\section{Theoretical Results}
We consider two general types of priors $\pi(\gamma_i)$ given by 
\begin{align}
\pi\left(\gamma_{i}\right)&=\gamma_{i}^{-a-1}L\left(\gamma_{i}\right),~a>0,\tag{P}\label{eq:polynomial}\\
\pi\left(\gamma_{i}\right)&=\exp\left(-b\gamma_{i}\right)L\left(\gamma_{i}\right),~b>0,\tag{E}\label{eq:exponential}
\end{align}
where $L\left(\cdot\right)$ is a nonnegative slowly varying function in Karamata's sense \citep[p.~6]{bingham1987regular} defined on $\left(0,\infty\right)$. We will call the priors in the form of \eqref{eq:polynomial} and \eqref{eq:exponential} polynomial-tailed priors and exponential-tailed priors respectively. As we will show later in this section that the theoretical performances of the HT method is closely related with the tails of the prior of $\gamma_i$. Table \ref{table:priors} provides a list of commonly used scale mixture priors of $\beta_i$ and the corresponding form of $\pi(\gamma_i)$. The Class column gives which class the prior belongs to. The corresponding form of $L$ is given in the last column. The half-hyperbolic and positive logistic distribution are included in the list as examples of exponential-tailed distributions other than the exponential distributions, although they have not been used as priors in literature. 

\renewcommand{\arraystretch}{2}
\begin{table}[ht]
\scriptsize
\centering
\begin{tabular}{cccc}
  \toprule 
Prior & $\pi(\gamma_i)/C$ & Class &  $L(\gamma_i)/C$  \\
  \midrule
Double-exponential & $\exp\{-b\gamma_i\}$ & E & 1 \\
Half-hyperbolic & $\exp\left\{-b\sqrt{1+\gamma_i^2}\right\}$ & E & $\exp\left\{b\gamma_i-b\sqrt{1+\gamma_i^2}\right\}$ \\
Positive Logistic & $\exp(b\gamma_i)\{1+\exp(b\gamma_i)\}^{-2}$ & E & $\exp(2b\gamma_i)\{1+\exp(b\gamma_i)\}^{-2}$\\
Student's T & $\gamma_i^{-a-1}\exp(-{a}/{\gamma_i})$ & P & $\exp(-a/\gamma_i)$ \\
Horseshoe & $\gamma_i^{-1/2}(1+\gamma_i)^{-1}$ & P & $\gamma_i / (1 + \gamma_i)$ \\
Horseshoe+ & $ \gamma_i^{-1/2}(\gamma_i-1)^{-1}\log(\gamma_i)$ & P & $\gamma_i (\gamma_i - 1)^{-1}\log(\gamma_i)$\\
NEG & $ \left(1+\gamma_i\right)^{-1-a} $ & P & $\{\gamma_i/(1+\gamma_i)\}^{a+1}$ \\
TPBN & $\gamma_i^{u-1}(1+\gamma_i)^{-a-u}$ & P & $\{\gamma_i/(1+\gamma_i)\}^{a+u}$ \\
GDP & $\int_0^\infty \frac{\lambda^2}{2} \exp\left(-\frac{\lambda^2\gamma_i}{2}\right)\lambda^{2a-1}\exp(-\eta\lambda)d\lambda$ & P & $\int_0^{\infty} t^{a} \exp(-t-\eta\sqrt{2t/\gamma_i})dt$ \\
HIB & {$\gamma_i^{u-1}(1+\gamma_i)^{-(a+u)}\exp\left\{-\frac{s}{1+\gamma_i}\right\}\left\{\phi^2+\frac{1-\phi^2}{1+\gamma_i}\right\}^{-1}$} & P &  {$\{\gamma_i/(1+\gamma_i)\}^{a+u}\exp\left\{-\frac{s}{1+\gamma_i}\right\}\left\{\phi^2+\frac{1-\phi^2}{1+\gamma_i}\right\}^{-1}$}\\
   \bottomrule
\end{tabular}

\caption{A list of scale mixture of normals shrinkage priors of $\beta_i$. $C$ is a generic constant. NEG: normal exponential gamma priors\citep{GB2005}, TPBN: three parameter beta normal priors \citep{armagan2011generalized}, GDP: generalized double Pareto priors \citep{ADL2012}, HIB: hypergeometric inverted beta priors \citep{polson2012half}.} 
\label{table:priors}
\end{table}
\renewcommand{\arraystretch}{1}

In this section, we will assume the design matrix is orthognal, that is $\mathbf{X}^T\mathbf{X}=n\mathbf{I}_p$. With this assumption, $$E(\beta_i \, | \, \gamma_i, \tau, \mathbf{Y}) = \frac{n\tau\gamma_i}{n\tau\gamma_i + 1} \hat{\beta}_i = (1-s_i)\hat{\beta}_i,$$ where $s_i = 1/(1+ n\tau\gamma_i)$, is the shrinkage factor. By law of iterated expectations, $$\hat{\beta}_i^{\text{PM}} = E(\beta_i \, | \, \mathbf{Y}) = (1-E(s_i \, | \, \mathbf{Y}))\hat{\beta}_i.$$ 
Therefore, with the orthogonal design matrix assumption, the selection criterion \eqref{eq:Criterion} of the proposed method simplifies to 
\begin{equation} \label{eq:Criterion_simple}
1-E(s_i \, | \, \mathbf{Y}) > 1/2.
\end{equation}
A similar procedure was considered by \citet{ghosh2015asymptotic} in the multiple testing context.

Following \citet{fan2001variable} and \citet{zou2006adaptive}, we say a variable selection procedure is oracle if it results in both variable selection consistency and optimal estimation rate. 
Let $\mathcal{A} = \{j: \beta_j^0 \neq 0\}$ and $\mathcal{A}_n = \{j: \hat{\beta}_j^{HT} \neq 0\}$. 
The variable selection consistency means $$\lim_{n \rightarrow \infty}P(\mathcal{A}_n = \mathcal{A})=1,~\text{as}~n\rightarrow \infty,$$ 
while the optimal estimation rate means $$\sqrt{n}(\hat{\bbeta}_{\mathcal{A}}^{{HT}}-\bbeta_{\mathcal{A}}^0) \overset{d}{\rightarrow} N(0, \sigma^2\mathbf{I}_{p_0}),~\text{as}~n\rightarrow\infty,$$ where $p_0$ is the cardinality of $\mathcal{A}$ and it does not depend on $n$.

Another thing to clarify before the presentation of our theoretical results is the treatment of the global parameter $\tau$. In \citet{datta2013asymptotic} and part of the results of \citet{ghosh2015asymptotic}, $\tau$ was treated as a tuning parameter. \citet{carvalho2010horseshoe} considered a full Bayesian treatment and a half-Cauchy prior for the global parameter. \citet{ghosh2015asymptotic} also provided some results when an empirical Bayes estimate of the global parameter is used. In this article, we treat $\tau$ as a tuning parameter or assume a hyper-prior for it. To distinguish the two treatments, we will write $\tau$ as $\tau_n$ when it is a tuning parameter.

\subsection{Properties of shrinkage factors}

By \eqref{eq:Criterion_simple}, the HT procedure is closely related with the shrinkage factor $s_i$, so we present its properties first.

\begin{proposition}\label{prop:s_consistency} Suppose the prior of $\gamma_i$ is proper. 
For $i \not\in \mathcal{A}$, if $n\tau_{n}\rightarrow0$, as $n\rightarrow\infty$, then 
$E(1 - s_{i} \, | \, \tau_n, \bY) \overset{p}{\rightarrow} 0$ as $n\rightarrow\infty$.
For $i \in \mathcal{A}$,
\begin{enumerate}
\item if $\gamma_i$ has a polynomial-tailed prior described in \eqref{eq:polynomial} and $n\tau_n \rightarrow 0$, $\log(\tau_n)/n \rightarrow 0$ as $n\rightarrow\infty$, then $E(1-s_i\, |\, \tau_n, \bY) \overset{p}{\rightarrow} 1$, as $n \rightarrow \infty$.
\item if $\gamma_i$ has an exponential-tailed prior described in \eqref{eq:exponential} and $n\tau_n \rightarrow 0$ and $n^2\tau_n \rightarrow \infty$ as $n\rightarrow\infty$, then $E(1 - s_i \, |\, \tau_n, \bY) \overset{p}{\rightarrow} 1$, as $n\rightarrow \infty$.
\end{enumerate}
\end{proposition}

Proposition \ref{prop:s_consistency} shows that, regardless of the choice of the prior of $\gamma_i$ in the given class, the HT procedure can identify an irrelevant variable correctly if $\tau_n$ goes to zero at a rate faster than $n^{-1}$. On the other hand, $\tau_n$ should not converge to zero too fast in order to avoid overshrinkage and to correctly identify relevant variables. The conditions $\log(\tau_n)/n \rightarrow 0$ and $n^2\tau_n \rightarrow \infty$ serve this idea for polynomial-tailed priors and exponential-tailed priors respectively. Given $n\tau_n \rightarrow 0$, the condition $n^2 \tau_n \rightarrow \infty$ is more stringent than $\log(\tau_n)/n \rightarrow 0$. Intuitively, this has to the case since exponential tails are lighter. To guarantee that the coefficients of important variables are not overly shrunk, the global parameter should decay at a slower rate and compensate the amount of shrinkage brought by exponential local parameters. 

\subsection{Polynomial-tailed priors}
Before presenting the main results of the HT procedure with polynomail-tailed priors, we would like to introduce an assumption that will be frequently mentioned in the rest of the section. We say a sequence of positive real numbers $\{t_n\}_{n = 1}^\infty$ satisfies poly-$a$ condition if there exists $\epsilon \in (0, a)$ such that
\begin{equation*}
p_n(nt_n)^\epsilon \rightarrow 0~\mbox{and}~\log ( t_n )/ \sqrt{n} \rightarrow 0,~\mbox{as}~n \rightarrow \infty.
\end{equation*}
Let $\tau_n = n^{-1-\frac{2}{a}}$. It satisfies the poly-$a$ condition since $p_n \leq n$. 
If $p$ does not vary with $n$, the condition can be simplified to $nt_n \rightarrow 0$ and $\log (t_n) / \sqrt{n} \rightarrow 0$.

\begin{theorem}\label{thm:tuning_poly}
Suppose a proper polynomial-tailed prior of the form \eqref{eq:polynomial} is assumed for $\gamma_i, i=1, \ldots, p$ with $0<a<1$. If $\{\tau_n\}_{n=1}^\infty$ satisfies the poly-$a$ condition, then the HT procedure is oracle.
\end{theorem}

As Theorem \ref{thm:tuning_poly} demenstrates, if $\tau_n$ is chosen to decay to zero at an appropriate rate, the HT procedure has the oracle property. This suggests if a hyperprior $\pi_n^{\tau}$ of $\tau$ concentrates most of its probability mass in an interval with its end points satisfying the poly-$a$ condition, then the HT threshold should still enjoy the oracle property. With this observation, we have the following result.
\begin{corollary}\label{thm:prior_poly}
Suppose that a proper polynomial-tailed prior of the form \eqref{eq:polynomial} is assumed for $\gamma_i, i=1, \ldots, p$ with $0<a<1$. We also place a prior $\pi_n^{\tau}$ with support $(\xi_n, \psi_n)$ on $\tau$. If both $\{\xi_n\}_{n=1}^\infty$ and $\{\psi_n\}_{n=1}^\infty$ satisfy the poly-$a$ condition, then the HT procedure is oracle.
\end{corollary}

\subsection{Exponential-tailed priors}
Now we examine the properties of HT procedure when exponential-tailed priors are assumed for the local parameters $\gamma_i$.
\begin{theorem}\label{thm:tuning_exp_vs}
Suppose a proper exponential-tailed prior of the form \eqref{eq:exponential} is assumed for $\gamma_i, i=1, \ldots, p$ and $\int_0^\infty \gamma_i \pi(\gamma_i) d\gamma_i < \infty$. If $n\tau_n \rightarrow 0$, $n^2\tau_n \rightarrow \infty$ and $\frac{p_n n \tau_n}{\sqrt{\log(n\tau_n)}} \rightarrow 0$, as $n\rightarrow \infty$, then the HT procedure achieves variable selection consistency.
\end{theorem}

\noindent
{\bf Remark.} Let $\tau_n = \log\log n/n^2$. It satisfies the conditions in Theorem \ref{thm:tuning_exp_vs}. Similar to what we have mentioned in Section 3.1, these conditions are more stringent than the poly-$a$ condition since exponential tails are lighter than polynomial ones. If $p$ does not depend on $n$, the conditions simplify to $n\tau_n \rightarrow 0$ and $n^2\tau_n \rightarrow \infty$, as $n \rightarrow \infty$. If $\pi(\gamma_i) = \exp(-\gamma_i)$, then the marginal prior of $\beta_i$ is proportional to $\exp(-|\beta_i|/\sqrt{\tau_n})$. Thus $1/\sqrt{\tau_n}$ corresponds to the penalty parameter $\lambda_n$ in the lasso estimator. 
The two conditions on $\tau_n$ assuming fixed $p$ can be translated to $\lambda_n/\sqrt{n} \rightarrow \infty$ and $\lambda_n/n \rightarrow 0$, which is a sufficient condition for the lasso estimator to be model selection consistent assuming the irrepresentable condition \citep{zou2006adaptive}.
\noindent

\begin{theorem}\label{thm:tuning_exp_an}
Suppose a proper exponential-tailed prior of the form \eqref{eq:exponential} and there exist $0 < m \leq M < \infty$ such that $m < L(t) < M$ for all $t \in (0, \infty)$. If $n\tau_n \rightarrow 0$ and $n^2\tau_n \rightarrow \infty$, then for $i \in \mathcal{A}$, with probability 1, 
$$\frac{m}{M}S^{(i)}_n \leq n\sqrt{\tau_{n}}\left(\hat{\beta}_{i}^{HT}-\beta_{i}^{0}\right) \leq \frac{M}{m}S^{(i)}_n,$$
where $\{S^{(i)}_n, n\geq 1\}$ are sequences of random variables and $S^{(i)}_n \overset{p}{\rightarrow} -\sqrt{2b}\sigma \sign(\beta_i^0)$.
\end{theorem}
\noindent
{\bf Remark.} Theorem \ref{thm:tuning_exp_vs} shows that, with exponential-tailed priors on local parameters, the HT procedure can achieve variable selection consistency when $\tau_n$ vanishes at certain rate. However, Theorem \ref{thm:tuning_exp_an} tells us that the procedure cannot achieve optimal estimation rate with $\tau_n$ decaying at this rate. The boundedness condition in Theorem \ref{thm:tuning_exp_an} looks restrictive, but all the three exponential-tailed distributions listed in Table \ref{table:priors} satisfy this condition. If $\beta_i$ has a double exponential prior as in \citet{park2008bayesian}, $L(t) = 1$, for all $t > 0$. In this case, we have $$n\sqrt{\tau_{n}}\left(\hat{\beta}_{i}^{HT}-\beta_{i}^{0}\right)\overset{p}{\rightarrow} -\sqrt{2b}\sigma \sign(\beta_i^0).$$

Next we will show the HT procedure with exponential-tailed priors does not have the oracle property either for other choice of $\tau_n$.
\begin{proposition}\label{prop:not_vs}
If the prior of $\gamma_i$, $\pi(\gamma_i)$, satisfies the condition
\begin{equation}\label{eq:tuning_exp_cond1}
\int_0^{\infty} \gamma_i^{-1/2}\pi(\gamma_i)d\gamma_i < \infty, \textrm{ for } i=1,\ldots,p,
\end{equation} 
then the HT procedure cannot achieve variable selection consistency when $n\tau_n \rightarrow c \in (0, \infty]$ as $n \rightarrow \infty$.
\end{proposition}

\noindent
{\bf Remark.} Proposition \ref{prop:not_vs} holds for both polynomial-tailed and exponential-tailed priors. The finite integral condition is not very restrictive for exponential-tailed priors of $\gamma_i$. In fact, the three exponential-tailed distributions in Table \ref{table:priors} satisfies the condition. However, it excludes horseshoe priors and some other priors in the polynomial-tailed prior class.

\begin{proposition}\label{prop:not_an}
Suppose $\pi(\gamma_i)$ is a proper exponential-tailed prior of the form \eqref{eq:exponential} and there exist $0 < m \leq M < \infty$ such that
\begin{equation}\label{eq:tuning_exp_cond2}
m < L(t) < M\textrm{ for all }t \in (0, \infty).
\end{equation}
If $n\tau_n \rightarrow 0$ as $n \rightarrow \infty$, the HT procedure cannot achieve optimal estimation rate. 
\end{proposition}
\noindent
{\bf Remark.}  The proofs of Theorem \ref{thm:tuning_exp} implies the HT procedure over-shinks the nonzero coefficients and the convergence rate is slower than $n^{1/2}$.

Combining the above propositions, we have the following theorem:
\begin{theorem}\label{thm:tuning_exp}
If $\pi(\gamma_i)$ is a proper exponential-tailed prior of the form \eqref{eq:exponential} and it satisfies conditions \eqref{eq:tuning_exp_cond1} and \eqref{eq:tuning_exp_cond2}, then the HT procedure does not have the oracle property for any choice of $\tau_n$.
\end{theorem}


\noindent
{\bf Remark.} As a special case, Theorem \ref{thm:tuning_exp} implies that the HT procedure lacks oracle if the prior introduced in \citet{park2008bayesian} is used. 
%

\section{Simulation Results}
In this section we apply the HT technique to the TPBN priors and the double-exponential (DE) priors, and compare them with the lasso, the
adaptive lasso, and the MLE in terms of prediction
performance and variable selection accuracy, when applicable. TPBN
priors are used with three different sets hyper-parameters: $a=0.5,u=0.5$,
which gives the horseshoe prior; $a=0.5,u=0.1$, which places more
probability mass around 0 than the horseshoe; $a=0.5,u=1$, which
results the normal-exponential-gamma (NEG) prior. We denote these 3 priors as TPBN-HS, TPBN-0.1, TPBN-NEG, respectively. We generate data
from $\bY=\mathbf{X}{\bbeta}+\boldsymbol{\epsilon}$, where $\boldsymbol{\epsilon}\sim{N}_{n}\left(0,\sigma^{2}\mathbf{I}_{n}\right)$.
To compare the prediction performance, we report the relative prediction
error (RPE) $E\left[\left(\hat{y}-\mathbf{x}^{T}{\bbeta}\right)^{2}\right]/\sigma^{2}$. To measure the variable selection accuracy, we use the misclassification error which is the proportion of variables incorrectly identified.

We use the LARS algorithm \citep{efron2004least} to fit both the
lasso and the adaptive lasso. The penalty parameter $\lambda$ is
chosen by 5-fold cross validation. Parameter $\gamma$ in the adaptive
lasso is fixed to be one so that the adaptive weights are the reciprocal
of the least square estimates. The Bayesian lasso and TPBN methods
are fit with the Gibbs sampler using the \texttt{rstan} package \citep{stan_development_team_stan:_2014}. The first example below was used in \citet{zou2006adaptive}.

Example 1 (A few large effects and a few 0s) We let 
\[
{\bbeta}=\left(3,1.5,0,0,2,0,0,0\right),
\]
the columns of $\mathbf{X}$ are normal vectors and the pairwise correlation
between $\mathbf{x}_{i}$ and $\mathbf{x}_{j}$ are 0.5 for $i\neq j$. We
let $\sigma=1,3,5$ and $n=20,50,80$ to compare the performance with
varying signal to noise ratios and sample sizes.

Example 2 (A few large effects, a few small effects, and a few 0s) We let
\[
\bbeta=\left(3,1.5,0.1,0.01,2,0,0,0\right),
\]
and the rest of the set up is the same as Example 1.

For each example, we generate 50 training and corresponding testing data sets. Each model is fit on the training set and the RPE is computed on the test set. We summarize the RPE results in Table \ref{table:RPE} and use the asterisk to denote the model with the smallest RPE. A couple of observations can be made from Table \ref{table:RPE}:
\begin{itemize}
 \item In general, the HT methods with TPBN priors are the best in prediction for both examples, especially with $u=0.1$ which gives the lowest RPE in 5 out of 9 cases.
 \item A smaller $u$ works better when the signal to noise ratio is low and the sample size is big,  while a larger $u$ results in better prediction performance when the signal to noise ratio is high.
 \item TPBN-0.1 predicts better than the adaptive lasso in almost all cases for both examples.
 \item The performance of the HT method with DE priors is comparable to the lasso in terms of prediction.
\end{itemize}

The mean misclassification errors are summarized in Table \ref{table:mis}. We also give the number of predictors chosen for each method in each set up in Table \ref{table:number_predictors}. TPBN-0.1 is the best model for variable selection in general. It has slightly lower misclassification error than the adaptive lasso and tends to be more parsimonious. The TPBN priors with larger $u$ including the horseshoe prior select significantly more variables and thus result in a much higher misclassification error due to less point mass around 0. The DE prior also tends to select most variables and does not perform well for variable selection, which agrees with our theoretical results. Finally, the lasso, which does not satisfy the oracle property, is not as accurate as the adaptive lasso and the TPBN-0.1.

\begin{table}[htb]
\centering
\caption{Median relative prediction error (RPE) for seven
methods in two simulation examples, based on 50 replications.} 
\label{table:RPE}
\small
\begin{tabular}{clllllllllll}
\hline
  & \multicolumn{3}{c}{$n=20$} & \phantom{0}& \multicolumn{3}{c}{$n=50$} & \phantom{0} & \multicolumn{3}{c}{$n=80$} \\
  \hline
  & $\sigma=1$ & $\sigma=3$ & $\sigma=5$ && $\sigma=1$ & $\sigma=3$ & $\sigma=5$ && $\sigma=1$ & $\sigma=3$ & $\sigma=5$ \\ 
\hline
Example\text{ }1\\
LS & 1.74 & 1.78 & 1.75 &  & 1.20 & 1.20 & 1.21 &  & 1.13 & 1.13 & 1.13 \\ 
  lasso & 1.54 & 1.50 & 1.51 &  & 1.16 & 1.14 & 1.15 &  & 1.12 & 1.11 & 1.10 \\ 
  adap lasso & 1.96 & 1.72 & 1.60 &  & 1.18 & 1.18 & 1.17 &  & 1.15 & 1.10 & 1.13 \\ 
  DE & 1.59 & 1.40 & 1.42* &  & 1.18 & 1.14 & 1.16 &  & 1.12 & 1.11 & 1.10 \\ 
  TPBN-HS & 1.52 & 1.38* & 1.46 &  & 1.14 & 1.12 & 1.13* &  & 1.11 & 1.08 & 1.08* \\ 
  TPBN-0.1 & 1.36* & 1.48 & 1.58 &  & 1.10* & 1.11* & 1.16 &  & 1.07* & 1.07* & 1.09 \\ 
  TPBN-NEG & 1.60 & 1.38* & 1.45 &  & 1.17 & 1.14 & 1.14 &  & 1.12 & 1.10 & 1.10 \\ 
  \hline
Example\text{ }2\\
LS & 1.60 & 1.62 & 1.90 &  & 1.22 & 1.22 & 1.24 &  & 1.09 & 1.11 & 1.12 \\ 
  lasso & 1.56 & 1.51 & 1.42* &  & 1.17 & 1.16 & 1.18 &  & 1.09 & 1.10 & 1.09* \\ 
  adap lasso & 1.61 & 1.99 & 1.67 &  & 1.20 & 1.21 & 1.19 &  & 1.08 & 1.11 & 1.12 \\ 
  DE & 1.52 & 1.53 & 1.47 &  & 1.19 & 1.16 & 1.16 &  & 1.09 & 1.10 & 1.09* \\ 
  TPBN-HS & 1.38 & 1.46* & 1.49 &  & 1.16 & 1.12 & 1.15* &  & 1.08 & 1.08 & 1.09* \\ 
  TPBN-0.1 & 1.34* & 1.73 & 1.57 &  & 1.13* & 1.11* & 1.21 &  & 1.07* & 1.07* & 1.11 \\ 
  TPBN-NEG & 1.47 & 1.50 & 1.42* &  & 1.19 & 1.15 & 1.15* &  & 1.09 & 1.09 & 1.09* \\ 
   \hline
   
\end{tabular}
\end{table}
\begin{table}[ht]
\centering
\caption{Mean misclassification error based on 50 replications.} 
\label{table:mis}
\small
\begin{tabular}{cccccccccccc}
\hline
  & \multicolumn{3}{c}{$n=20$} & \phantom{0}& \multicolumn{3}{c}{$n=50$} & \phantom{0} & \multicolumn{3}{c}{$n=80$} \\
  \hline
  & $\sigma=1$ & $\sigma=3$ & $\sigma=5$ && $\sigma=1$ & $\sigma=3$ & $\sigma=5$ && $\sigma=1$ & $\sigma=3$ & $\sigma=5$ \\ 
  \hline
  Example\text{ }1\\
lasso & 0.30 & 0.26 & 0.24 &  & 0.31 & 0.25 & 0.28 &  & 0.37 & 0.30 & 0.32 \\ 
  adap lasso & 0.16 & 0.19 & 0.28 &  & 0.19 & 0.16 & 0.17 &  & 0.17 & 0.18 & 0.16 \\ 
  DE & 0.53 & 0.37 & 0.34 &  & 0.60 & 0.47 & 0.43 &  & 0.59 & 0.56 & 0.49 \\ 
  TPBN-HS & 0.41 & 0.25 & 0.27 &  & 0.51 & 0.27 & 0.25 &  & 0.55 & 0.37 & 0.32 \\ 
  TPBN-0.1 & 0.12 & 0.19 & 0.28 &  & 0.12 & 0.09 & 0.18 &  & 0.10 & 0.08 & 0.16 \\ 
  TPBN-NEG & 0.52 & 0.33 & 0.30 &  & 0.60 & 0.47 & 0.39 &  & 0.59 & 0.57 & 0.47 \\ 
  \hline
  Example\text{ }2\\
  lasso & 0.41 & 0.30 & 0.37 &  & 0.37 & 0.40 & 0.39 &  & 0.35 & 0.37 & 0.33 \\ 
  adap lasso & 0.26 & 0.27 & 0.33 &  & 0.28 & 0.27 & 0.29 &  & 0.25 & 0.24 & 0.26 \\ 
  DE & 0.57 & 0.44 & 0.43 &  & 0.60 & 0.50 & 0.45 &  & 0.61 & 0.60 & 0.50 \\ 
  TPBN-HS & 0.46 & 0.35 & 0.36 &  & 0.53 & 0.36 & 0.36 &  & 0.55 & 0.49 & 0.37 \\ 
  TPBN-0.1 & 0.24 & 0.29 & 0.35 &  & 0.24 & 0.19 & 0.28 &  & 0.22 & 0.25 & 0.27 \\ 
  TPBN-NEG & 0.57 & 0.43 & 0.40 &  & 0.61 & 0.51 & 0.42 &  & 0.61 & 0.60 & 0.48 \\ 
\hline
\end{tabular}
\end{table}
\begin{table}[ht]

\centering
\caption{Mean number of predictors selected for six
methods in two examples, based on 50 replications.} 
\label{table:number_predictors}
\small
\begin{tabular}{cccccccccccc}
  \hline
  & \multicolumn{3}{c}{$n=20$} & \phantom{0}& \multicolumn{3}{c}{$n=50$} & \phantom{0} & \multicolumn{3}{c}{$n=80$} \\
  & $\sigma=1$ & $\sigma=3$ & $\sigma=5$ && $\sigma=1$ & $\sigma=3$ & $\sigma=5$ && $\sigma=1$ & $\sigma=3$ & $\sigma=5$ \\ 
\hline
  Example\text{ }1\\
lasso & 5.40 & 4.94 & 3.72 &  & 5.48 & 5.00 & 5.02 &  & 5.94 & 5.36 & 5.52 \\ 
  adap lasso & 3.92 & 3.42 & 2.62 &  & 4.24 & 3.82 & 3.50 &  & 4.20 & 4.24 & 3.90 \\ 
  DE & 7.20 & 5.62 & 4.72 &  & 7.80 & 6.78 & 6.38 &  & 7.74 & 7.44 & 6.96 \\ 
  TPBN-HS & 6.28 & 4.56 & 3.54 &  & 7.08 & 5.18 & 4.74 &  & 7.40 & 5.98 & 5.56 \\ 
  TPBN-0.1 & 3.98 & 3.20 & 2.22 &  & 4.00 & 3.44 & 3.06 &  & 3.80 & 3.60 & 3.56 \\ 
  TPBN-NEG & 7.18 & 5.38 & 4.52 &  & 7.80 & 6.80 & 6.06 &  & 7.76 & 7.58 & 6.78 \\ 
  \hline
  Example\text{ }2\\
lasso & 4.67 & 4.00 & 3.79 &  & 5.10 & 5.13 & 4.70 &  & 4.85 & 4.58 & 4.14 \\ 
  adap lasso & 2.81 & 2.15 & 2.22 &  & 3.26 & 3.06 & 3.18 &  & 3.17 & 2.88 & 2.65 \\ 
  DE & 7.12 & 5.66 & 4.30 &  & 7.71 & 6.57 & 5.64 &  & 7.80 & 7.36 & 6.31 \\ 
  TPBN-HS & 5.86 & 4.28 & 2.90 &  & 6.82 & 4.55 & 4.22 &  & 7.13 & 5.84 & 4.46 \\ 
  TPBN-0.1 & 3.06 & 2.68 & 1.81 &  & 2.99 & 2.27 & 2.35 &  & 2.91 & 2.90 & 2.81 \\ 
  TPBN-NEG & 7.15 & 5.50 & 4.11 &  & 7.75 & 6.48 & 5.24 &  & 7.80 & 7.45 & 6.08 \\ 
\hline
\end{tabular}

\end{table}

\section{Discussion}
In this paper, we consider Bayesian variable selection problem of linear regression model with global-local shrinkage priors on the regression coefficients. Our proposed variable selection procedure selects a variable if the ratio of the posterior mean to the ordinary least square estimate of the corresponding coefficient is greater than $1/2$. With orthogonal design matrices, we show that if the local parameters have polynomial-tailed priors, our proposed method is oracle in the sense that it can achieve variable selection consistency and optimal estimation rate at the same time. However, if, instead, an exponential-tailed prior is used for the local parameters, the proposed method does not have the oracle property.

Although the theoretical results are obtained only under the assumption of orthogonal designs, the simulation study shows the performance of our method is similar when applied to design matrices with moderate correlation. 

Because of the use of the ordinary least square estimate in the proposed method, we only consider the situation when the sample size is greater than the number of predictors in the model. In the case that $p > n$, the inverse of $\bX^T\bX$ does not exist. One possible way to get around this is to use a generalized inverse, say the Moore-Penrose generalized inverse. Since generalized inverse is not unique, it is critical to examine if different choices can produce consistent selection results or if certain choice of the generalized inverse will outperform others.

\appendix
\section{Technical proofs of results in Section 3}
\begin{lemma}[Properties of slowly varying functions]\label{lemma:svf} If $L$ is a slowly varying function, then
\begin{enumerate}[{(i)}]
\item $L^\alpha$ is slowly varying for all $\alpha \in \mathbb{R}$.\label{eq:svf_1}
\item ${\log L(x)}/{\log x} \rightarrow 0$, as $x\rightarrow \infty$.\label{eq:svf_2}
\item $x^{-\alpha} L(x) \rightarrow 0$ and $x^\alpha L(x) \rightarrow \infty$, as $x \rightarrow \infty$ for all $\alpha > 0$.\label{eq:svf_3}
\item for $\alpha < -1$, $-\frac{\int_x^\infty t^\alpha L(t) dt}{x^{\alpha+1} /(\alpha+1)} \rightarrow 1$, as $x \rightarrow \infty$.\label{eq:svf_4}
\item there exist $A_0 > 0$ such that for $\alpha > -1$, $\frac{\int_{A_0}^x t^\alpha L(t) dt}{x^{\alpha+1} /(\alpha+1)} \rightarrow 1$, as $x \rightarrow \infty$.\label{eq:svf_5}
\end{enumerate}
\end{lemma}

\begin{proof}
See Propositions 1.3.6, 1.5.8 and 1.5.10 of \citet{bingham1987regular}.
\end{proof}

\begin{lemma} \label{lemma:noise_2}
Suppose $n\tau_n \rightarrow 0$, as $n \rightarrow \infty$.
\begin{enumerate}
\item If $\gamma_i$ has a proper polynomial-tailed prior described in \eqref{eq:polynomial} with $0 < a < 1$, then there exist $A_0 > 1$ such that
$$E(1 - s_i \, | \, \tau_n, \mathbf{Y}) \leq \frac{A_0 (n\tau_n)^a}{a(1-a)} L\left(\frac{1}{n\tau_n}\right)\exp\left(\frac{n\hat{\beta}_i}{2\sigma^2}\right)(1+o(1))$$
\item If $\gamma_i$ has a proper exponential-tailed prior described in \eqref{eq:exponential} and $C = \int_0^{\infty} \gamma_i \pi(\gamma_i) d\gamma_i < \infty$, then
$$E(1 - s_i \, | \, \tau_n, \mathbf{Y}) \leq Cn\tau_n \exp\left(\frac{n\hat{\beta}_i^2}{2\sigma^2}\right)(1 + o(1)).$$
\end{enumerate}
The $o(1)$ terms in both cases do not depend on $i$.
\end{lemma}

\begin{proof}
By Lebesgue dominated convergence Theorem,
\begin{align*}
E(1-s_{i} \, | \, \tau_n, \bY) &=  \frac{\int_{0}^{\infty}\frac{n\tau_{n}\gamma_{i}}{1+n\tau_{n}\gamma_{i}}\left(1+n\tau_{n}\gamma_{i}\right)^{-\frac{1}{2}}\exp\left\{ \frac{n\hat{\beta}_{i}^{2}}{2\sigma^{2}}\frac{n\tau_{n}\gamma_{i}}{1+n\tau_{n}\gamma_{i}}\right\} \pi\left(\gamma_{i}\right)d\gamma_{i}}{\int_{0}^{\infty}\left(1+n\tau_{n}\gamma_{i}\right)^{-\frac{1}{2}}\exp\left\{ \frac{n\hat{\beta}_{i}^{2}}{2\sigma^{2}}\frac{n\tau_{n}\gamma_{i}}{1+n\tau_{n}\gamma_{i}}\right\} \pi\left(\gamma_{i}\right)d\gamma_{i}}\\
& \leq  \frac{\int_{0}^{\infty}\left(n\tau_{n}\gamma_{i}\right)\left(1+n\tau_{n}\gamma_{i}\right)^{-\frac{3}{2}}\pi\left(\gamma_{i}\right)d\gamma_{i}}{\int_{0}^{\infty}\left(1+n\tau_{n}\gamma_{i}\right)^{-\frac{1}{2}}\pi\left(\gamma_{i}\right)d\gamma_{i}}\exp\left\{ \frac{n\hat{\beta}_{i}^{2}}{2\sigma^{2}}\right\}\\
& = \int_{0}^{\infty}\left(n\tau_{n}\gamma_{i}\right)\left(1+n\tau_{n}\gamma_{i}\right)^{-\frac{3}{2}}\pi\left(\gamma_{i}\right)d\gamma_{i}\exp\left\{ \frac{n\hat{\beta}_{i}^{2}}{2\sigma^{2}}\right\} (1+o(1)).
\end{align*}

We now consider the case that $\gamma_i$ has a prior from the polynomial class. 
By property \eqref{eq:svf_5} in Lemma \ref{lemma:svf}, there exists $A_0 \geq 1$ such that $\frac{\int_{A_0}^x t^{-aL(t)dt}}{x^{1-a}L(x)} \rightarrow \frac{1}{1-a}$ as $x \rightarrow \infty$. Therefore,
\begin{align*}
& \int_{A_0}^{\frac{A_0}{n\tau_n}} \left(n\tau_{n}\gamma_{i}\right)\left(1+n\tau_{n}\gamma_{i}\right)^{-\frac{3}{2}}\gamma_i^{-a-1}L(\gamma_i)d\gamma_i \\
\leq &~ n\tau_n \int_{A_0}^{\infty} \gamma_i^{-a} L(\gamma_i) d\gamma_i = \frac{n\tau_n}{1-a}\left(\frac{A_0}{n\tau_n}\right)^{1-a}L\left(\frac{A_0}{n\tau_n}\right)(1+o(1)) \leq ~\frac{A_0}{1-a}(n\tau_n)^a L\left(\frac{1}{n\tau_n}\right)(1+o(1)).
\end{align*}
Also, $$\int_0^{A_0} \left(n\tau_{n}\gamma_{i}\right)\left(1+n\tau_{n}\gamma_{i}\right)^{-\frac{3}{2}}\gamma_i^{-a-1}L(\gamma_i)d\gamma_i \leq A_0 n\tau_n \int_0^{\infty} \gamma_i^{-a-1}L(\gamma_i)d\gamma_i = A_0 n \tau_n,$$ and
\begin{align*}
&\int_{\frac{A_0}{n\tau_n}}^{\infty} \left(n\tau_{n}\gamma_{i}\right)\left(1+n\tau_{n}\gamma_{i}\right)^{-\frac{3}{2}}\gamma_i^{-a-1}L(\gamma_i)d\gamma_i \\
\leq & ~ \int_{\frac{A_0}{n\tau_n}}^{\infty} \gamma_i^{-a-1} L(\gamma_i) d\gamma_i = \frac{1}{a}\left(\frac{A_0}{n\tau_n}\right)^{-a} L\left(\frac{1}{n\tau_n}\right)(1+o(1)) \leq ~ \frac{A_0}{a}(n\tau_n)^a L\left(\frac{1}{n\tau_n}\right) (1+o(1)).
\end{align*}
Hence, 
\begin{align*}
&\int_{0}^{\infty}\left(n\tau_{n}\gamma_{i}\right)\left(1+n\tau_{n}\gamma_{i}\right)^{-\frac{3}{2}}\gamma_i^{-a-1}L(\gamma_i)d\gamma_{i}\\
\leq &~ \frac{A_0 (n\tau_n)^a}{a(1-a)} L\left(\frac{1}{n\tau_n}\right) \left[\frac{(n\tau_n)^{1-a}a(1-a)}{L(1/(n\tau_n))} + a(1+o(1)) + (1-a)(1+o(1))\right]\\
= & ~ \frac{A_0(n\tau_n)^a}{a(1-a)} L\left(\frac{1}{n\tau_n}\right)(1 + o(1)).
\end{align*}

If $\gamma_i$ has a prior in the exponential-tailed class and $\int_0^\infty \gamma_i \pi(\gamma_i) d\gamma_i < \infty$, then 
\begin{align*}
\int_{0}^\infty n\tau_n\gamma_i(1+n\tau_n\gamma_i)^{-3/2} \pi(\gamma_i) d\gamma_i \leq n\tau_n \int_0^{\infty}\gamma_i \pi(\gamma_i)d\gamma_i = Cn\tau_n.
\end{align*}

\end{proof}

\begin{lemma}\label{lemma:signal_2}
Suppose $n\tau_n \rightarrow 0$ as $n \rightarrow \infty$ and $\eta$, $q$ are arbitrary constants in $(0,1)$.
\begin{enumerate}
\item If $\gamma_i$ has a proper polynomial-tailed prior described in \eqref{eq:polynomial}, then
$$P\left( s_i > \eta \, | \tau_n, \bY\right) \leq \frac{(a+\frac{1}{2})(\eta q)^{-a-\frac{1}{2}}(1-\eta q)^a}{(n\tau_n)^{a}L\left(\frac{1}{n\tau_n}(\frac{1}{\eta q}-1)\right)} \exp\left\{-\frac{n\hat{\beta}_i^2}{2\sigma^2}\eta(1-q)\right\} (1+o(1)).$$
\item If $\gamma_i$ has a proper exponential prior described in \eqref{eq:exponential}, then for sufficient large $n$ (not depending on $i$), 
$$P\left(s_{i}>\eta \, | \, \tau_n, \bY\right)\leq 2b\left(\frac{n\tau_n}{1-\eta q}\right)^\frac{1}{2}\exp\left\{ \frac{2b}{n\tau_n}\left(\frac{1}{\eta q}-1\right)\right\} \exp\left\{ -\frac{n\hat{\beta}_{i}^{2}}{2\sigma^{2}}\eta\left(1-q\right)\right\}.$$
\end{enumerate}
\end{lemma}

\begin{proof}
For any $\eta,q\in\left(0,1\right)$, 
\begin{align}
P\left(s_{i}>\eta \, |\, \tau_n, \bY\right) & =  P\left(\gamma_{i}<\frac{1}{n\tau_{n}}\left(\frac{1}{\eta}-1 \right)\,\middle | \, \tau_n, \bY\right)\nonumber\\
 & \leq  \frac{\int_{0}^{\frac{1}{n\tau_{n}}\left(\frac{1}{\eta}-1\right)}\left(1+n\tau_{n}\gamma_{i}\right)^{-\frac{1}{2}}\exp\left\{ -\frac{n\hat{\beta}_{i}^{2}}{2\sigma^{2}}\cdot\frac{1}{1+n\tau_{n}\gamma_{i}}\right\} \pi\left(\gamma_{i}\right)d\gamma_{i}}{\int_{\frac{1}{n\tau_{n}}\left(\frac{1}{\eta q}-1\right)}^{\infty}\left(1+n\tau_{n}\gamma_{i}\right)^{-\frac{1}{2}}\exp\left\{ -\frac{n\hat{\beta}_{i}^{2}}{2\sigma^{2}}\cdot\frac{1}{1+n\tau_{n}\gamma_{i}}\right\} \pi\left(\gamma_{i}\right)d\gamma_{i}} \nonumber \\
 & \leq  \frac{\int_{0}^{\frac{1}{n\tau_{n}}\left(\frac{1}{\eta}-1\right)}\left(1+n\tau_{n}\gamma_{i}\right)^{-\frac{1}{2}}\pi\left(\gamma_{i}\right)d\gamma_{i}}{\int_{\frac{1}{n\tau_{n}}\left(\frac{1}{\eta q}-1\right)}^{\infty}\left(1+n\tau_{n}\gamma_{i}\right)^{-\frac{1}{2}}\pi\left(\gamma_{i}\right)d\gamma_{i}}\exp\left\{ -\frac{n\hat{\beta}_{i}^{2}}{2\sigma^{2}}\eta(1-q)\right\}.
 \label{eq:RHS1}
\end{align}
The numerator of the first factor in \eqref{eq:RHS1} is bounded by 1. For the denominator (denoted by $D$), we discuss
the two types of priors separately.

First consider the case that $\gamma_i$ has a proper polynomial-tailed prior.
By property \eqref{eq:svf_4} of Lemma \ref{lemma:svf}, $$\frac{\left(\frac{1}{n\tau_n}(\frac{1}{\eta q} - 1)\right)^{-a-\frac{1}{2}}L\left(\frac{1}{n\tau_n}(\frac{1}{\eta q}-1)\right)}{\int_{\frac{1}{n\tau_n}(\frac{1}{\eta q}-1)}^{\infty}\gamma_i^{-a-\frac{3}{2}} L(\gamma_i)d\gamma_i} \rightarrow a+\frac{1}{2},~\text{as}~n \rightarrow \infty.$$ Hence
\begin{align*}
D & \geq \left(\frac{1-\eta q}{n\tau_n}\right)^{\frac{1}{2}} \frac{L\left(\frac{1}{n\tau_n}(\frac{1}{\eta q}-1)\right)}{(a+\frac{1}{2})\left(\frac{1}{n\tau_n}(\frac{1}{\eta q}-1)\right)^{a+\frac{1}{2}}}(1+o(1))\\
& = \frac{(n\tau_n)^a}{a+1/2} (\eta q)^{a+\frac{1}{2}}(1-\eta q)^{-a} L\left(\frac{1}{n\tau_n}(\frac{1}{\eta q}-1)\right)(1+o(1)).
\end{align*}

If $\gamma_i$ has a proper exponential-tailed prior, 
\begin{align*}
D & = \int_{\frac{1}{n\tau_{n}}\left(\frac{1}{\eta q}-1\right)}^{\infty}\left(1+n\tau_{n}\gamma_{i}\right)^{-\frac{1}{2}}\exp\left\{ -b\gamma_{i}\right\} L\left(\gamma_{i}\right)d\gamma_{i}\\
 & = \int_{\frac{1}{n\tau_{n}}\left(\frac{1}{\eta q}-1\right)}^{\infty}\left(\frac{n\tau_{n}\gamma_{i}}{1+n\tau_{n}\gamma_{i}}\right)^{\frac{1}{2}}\left(n\tau_{n}\right)^{-\frac{1}{2}}\gamma_{i}^{-\frac{1}{2}}\exp\left\{ -b\gamma_{i}\right\} L\left(\gamma_{i}\right)d\gamma_{i}\\
 & \geq \int_{\frac{1}{n\tau_{n}}\left(\frac{1}{\eta q}-1\right)}^{\infty}\left(1-\eta q\right)^{\frac{1}{2}}\left(n\tau_{n}\right)^{-\frac{1}{2}}\gamma_{i}^{-\frac{1}{2}}\exp\left\{ -b\gamma_{i}\right\} L\left(\gamma_{i}\right)d\gamma_{i}\\
 & = \int_{\frac{1}{n\tau_{n}}\left(\frac{1}{\eta q}-1\right)}^{\infty}\left(1-\eta q\right)^{\frac{1}{2}}\left(n\tau_{n}\right)^{-\frac{1}{2}}\exp\left\{ -2b\gamma_{i}\right\} \left(\exp\left\{ b\gamma_{i}\right\}\gamma_{i}^{-1}\right)\left(\gamma_{i}^{1/2}L\left(\gamma_{i}\right)\right)d\gamma_{i}
\end{align*}
Since $\exp(b\gamma_i)\gamma_i^{-1} \rightarrow \infty$ and $\gamma_i^{1/2}L(\gamma_i) \rightarrow \infty$ as $\gamma_i\rightarrow \infty$, for sufficiently large $n$,
\begin{align*}
D &\geq \int_{\frac{1}{n\tau_{n}}\left(\frac{1}{\eta q}-1\right)}^{\infty} \left(\frac{1-\eta q}{n\tau_{n}}\right)^{\frac{1}{2}}\exp\left\{ -2b\gamma_{i}\right\} d\gamma_{i} = \frac{1}{2b}\left(\frac{1-\eta q}{n\tau_n}\right)^{\frac{1}{2}}\exp\left\{ -\frac{2b}{n\tau_{n}}\left(\frac{1}{\eta q}-1\right)\right\}.
\end{align*}
\end{proof}

\begin{proof}[{\bf Proof of Proposition \ref{prop:s_consistency}}]
It is clear that
\begin{align}
E\left(1-s_{i}|\bY\right) & =  \frac{\int_{0}^{\infty}\frac{n\tau_{n}\gamma_{i}}{1+n\tau_{n}\gamma_{i}}\left(1+n\tau_{n}\gamma_{i}\right)^{-\frac{1}{2}}\exp\left\{ \frac{n\hat{\beta}_{i}^{2}}{2\sigma^{2}}\frac{n\tau_{n}\gamma_{i}}{1+n\tau_{n}\gamma_{i}}\right\} \pi\left(\gamma_{i}\right)d\gamma_{i}}{\int_{0}^{\infty}\left(1+n\tau_{n}\gamma_{i}\right)^{-\frac{1}{2}}\exp\left\{ \frac{n\hat{\beta}_{i}^{2}}{2\sigma^{2}}\frac{n\tau_{n}\gamma_{i}}{1+n\tau_{n}\gamma_{i}}\right\} \pi\left(\gamma_{i}\right)d\gamma_{i}} \nonumber\\
& \leq \frac{\int_{0}^{\infty}\left(n\tau_{n}\gamma_{i}\right)\left(1+n\tau_{n}\gamma_{i}\right)^{-\frac{3}{2}}\pi\left(\gamma_{i}\right)d\gamma_{i}}{\int_{0}^{\infty}\left(1+n\tau_{n}\gamma_{i}\right)^{-\frac{1}{2}}\pi\left(\gamma_{i}\right)d\gamma_{i}}\exp\left\{ \frac{n\hat{\beta}_{i}^{2}}{2\sigma^{2}}\right\}.\label{eq:1}
\end{align}
By Lebesgue  dominated convergence Theorem, the numerator and denominator in \eqref{eq:1} converges to 0 and 1 respectively as $n\rightarrow \infty$. If $i \not\in \mathcal{A}$, $n\hat\beta_i^2 = O_p(1)$. Therefore, $E(1-s_i \, | \, \tau_n, \bY) \overset{p}{\rightarrow} 0$, as $n \rightarrow \infty$ by Slutsky's theorem.

For any $0 <\epsilon \leq 1$,
$E(s_i \, | \, \tau_n, \bY) = \int_{0}^{\epsilon/2} s_i \,p(s_i \, | \, \tau_n, \bY) ds_i + \int_{\epsilon/2}^{1} s_i \,p(s_i \, | \, \tau_n, \bY) ds_i \leq {\epsilon}/{2} + P(s_i > \epsilon/2 \, | \, \tau_n, \bY).$ Thus 
$P\left(E(s_i \, | \, \tau_n, \bY) \geq \epsilon\right) \leq P\left(P(s_i > {\epsilon}/{2} \, | \, \tau_n, \bY) \geq {\epsilon}/{2}\right).$
If $\gamma_i$ has a polynomial-tailed prior, using the first part of Lemma \ref{lemma:signal_2} with $\eta=\epsilon/2$, the above inequality yields 
\begin{align*}
P(E(s_i \, | \, \tau_n, \bY) \geq \epsilon) &\leq P\left( \frac{(a+\frac{1}{2})(\eta q)^{-a-\frac{1}{2}}(1-\eta q)^a}{(n\tau_n)^{a}L\left(\frac{1}{n\tau_n}(\frac{1}{\eta q}-1)\right)} \exp\left\{-\frac{n\hat{\beta}_i^2}{2\sigma^2}\eta(1-q)\right\} > \epsilon/2\right) (1+o(1))\\
& = P\left(\hat{\beta}_i^2 < \frac{2\sigma^2}{\eta q}\left[\frac{c_1}{n} - \frac{a\log(n\tau_n)}{n}\left\{1+\frac{\log L(\frac{1}{n\tau_n}(\frac{1}{\eta q}-1))}{a\log(n\tau_n)}\right\}\right]\right)(1+o(1)),
\end{align*}
where $c_1$ is a constant that does not depend on $n$.
By property \eqref{eq:svf_2} in Lemma \ref{lemma:svf} and our assumptions, the terms in the bracket converge to zero as $n \rightarrow \infty$. Since $\hat{\beta}_i \overset{p}{\rightarrow} \beta_i^{0} \neq 0$, we have $P(E(s_i \, | \, \tau_n, \bY) \geq \epsilon) \rightarrow 0$, as $n\rightarrow \infty$.

If $\gamma_i$ has an exponential-tailed prior, by the second part of Lemma \ref{lemma:signal_2}, the assumption that $n\tau_{n}\rightarrow0$, $n^{2}\tau_{n}\rightarrow\infty$ as
$n\rightarrow\infty$ implies that $P\left(s_{i}>\eta \, | \, \tau_n, \bY\right)\stackrel{p}{\rightarrow}0$
for any $\eta>0$. Therefore $P(P(s_i > {\epsilon}/{2} \, | \, \tau_n, \bY) \geq {\epsilon}/{2}) \rightarrow 0$ and hence $P(E(s_i \, | \, \tau_n, \bY) \geq \epsilon) \rightarrow 0$, as $n \rightarrow \infty$.
\end{proof}

\begin{proof}[{\bf Proof of Theorem \ref{thm:tuning_poly}}]
We first prove the variable selection consistency part. It is clear that
\begin{equation*}
P(\mathcal{A}_n\neq \mathcal{A}) \leq 
\sum_{i\,\not\in\, \mathcal{A}} P\left(E(1 - s_i \, |\, \tau_n, \bY) \geq \frac{1}{2}\right) +  \sum_{i \, \in \, \mathcal{A}}P\left(E(1-s_i \, | \, \tau_n, \bY) < \frac{1}{2}\right).
\end{equation*}
Since $p_0 = |\mathcal{A}|$ does not depend on n, by Proposition \ref{prop:s_consistency}, the second term on the right hand side of the above inequality goes to zero as $n\rightarrow \infty$.
If $i \not\in \mathcal{A}$, by Lemma \ref{lemma:noise_2} and the fact that $\sqrt {n}\hat{\beta}_i$ has a standard normal distribution,   
$$P\left(E(1 - s_i \, |\, \tau_n, \bY) > \frac{1}{2}\right) \leq P\left(\exp\left(\frac{n\hat{\beta}_i^2}{2}\right)\frac{A_0(n\tau)^a}{a(1-a)}L\left(\frac{1}{n\tau}\right)\xi_n > 1/2\right) = 2\left[1 - \Phi (\sqrt{M_n})\right],$$ where $\xi_n$, not depending on $i$, is a generic term that converges to 1 as $n \rightarrow \infty$, and $M_n = 2\log\left(\frac{C}{(n\tau)^aL(1/(n\tau))\xi_n}\right)$ with $C$ being a generic constant. Noticing that the right hand side of the above inequality does not depend on $i$, $\sum_{i\,\not\in\, \mathcal{A}} P\left(E(1 - s_i \, |\, \tau_n, \bY) \geq \frac{1}{2}\right) \leq p_n P\left(E(1 - s_i \, |\, \tau_n, \bY) > \frac{1}{2}\right)$. Therefore, the proof of the variable selection consistency part will be complete if we can show $2p_n\left[1 - \Phi (\sqrt{M_n})\right]$ converge to zero as $n \rightarrow \infty$. In fact, by property \eqref{eq:svf_3} in Lemma \ref{lemma:svf}, $M_n \rightarrow \infty$, so
\begin{equation*}
2p_n\left[1 - \Phi (\sqrt{M_n})\right] \leq \frac{2\phi(\sqrt{M_n})}{\sqrt{M_n}} =Cp_n(n\tau_n)^\epsilon \frac{(n\tau_n)^{a-\epsilon}L(1/(n\tau_n))}{\sqrt{\log(1/(n\tau_n))}}(1 + o(1)).
\end{equation*}
Again by property \eqref{eq:svf_3} in Lemma \ref{lemma:svf}, $(n\tau_n)^{a-\epsilon} L(1/(n\tau_n))\rightarrow 0$ as $n\rightarrow \infty$. Therefore $2p_n\left[1 - \Phi (\sqrt{M_n})\right] \rightarrow 0$ if $p_n(n\tau_n)^\epsilon \rightarrow$, as $n\rightarrow \infty$.

Now we show the asymptotic normality part. For any
$i\in\mathcal{A}$, we have $\hat{\beta}_{i}\overset{p}{\rightarrow}\beta_{i}^{0}\neq0$,
$\sqrt{n}\left(\hat{\beta}_{i}-\beta_{i}^{0}\right)\stackrel{d}{\rightarrow}N\left(0,\sigma^{2}\right)$
and
$$
\sqrt{n}\left(\hat{\beta}_{i}^{HT}-\beta_{i}^{0}\right)=\sqrt{n}\left(\hat{\beta}_{i}-\beta_{i}^{0}\right)-\sqrt{n}E(s_{i} \, | \, \tau_n, \bY) \hat{\beta}_{i}-\sqrt{n}\hat{\beta}_i^{PM}I\left(E(1-s_i \, | \, \tau_n, \bY)\leq {1}/{2}\right).
$$
Since the third term on the right hand side converges to zero in probability by Proposition \ref{prop:s_consistency}, the proof of the asymptotic normality part will be complete if we can show that $\sqrt{n}E(s_i \, | \, \tau_n, \bY)$ converge to zero in probability. 
In fact, for any $\epsilon > 0$, by similar arguments as in the proof of Proposition \ref{prop:s_consistency}, 
$$P( \sqrt{n}E(s_i \, | \, \tau_n, \bY)  \geq \epsilon) \leq P(P(s_i \geq \epsilon/(2\sqrt{n}) \,|\, \tau_n, \bY) > \epsilon/(2\sqrt{n})).$$
In Lemma \ref{lemma:signal_2}, let $\eta = \eta_n = {\epsilon}/({2\sqrt{n}})$. Then
\begin{align*}
P\left(P\left(s_i \geq \frac{\epsilon}{2\sqrt{n}} \, \middle | \, \tau_n, \bY\right)>\frac{\epsilon}{2\sqrt{n}}\right) &\leq P\left(\frac{\left(a+\frac{1}{2}\right)\left(1-\frac{\epsilon q}{2\sqrt{n}}\right)^a\exp\left(-\frac{n\hat{\beta}_i^2\epsilon(1-q)}{4\sigma^2\sqrt{n}}\right)}{(n\tau_n)^a\left(\frac{\epsilon q}{2\sqrt{n}}\right)^{a+\frac{1}{2}} L\left(\frac{1}{n\tau_n}\left(\frac{2\sqrt{n}}{\epsilon q}-1\right)\right)} > \frac{\epsilon}{2\sqrt{n}}\right)(1+o(1)) \\
& = P(\hat{\beta}_i^2 < c_n)(1+o(1)),
\end{align*}
where $c_n = d_2 n^{-1/2}\left\{ \log\left(d_1n^{3/4}\right) + a\log\left(\frac{1}{n\tau_n}\left(\frac{2\sqrt{n}}{\epsilon q}-1\right)\right)\left[1-\frac{\log L\left(\frac{1}{n\tau_n}\left(\frac{2\sqrt{n}}{\epsilon q}-1\right)\right)}{a\log \left(\frac{1}{n\tau_n}\left(\frac{2\sqrt{n}}{\epsilon q}-1\right)\right)}\right] \right\}$. Since $c_n \rightarrow 0$ and $\hat{\beta}_i \overset{p}{\rightarrow}\beta_i^0 \neq 0$, we have $P(E(s_i \, | \, \tau_n, \bY) \geq \epsilon/\sqrt{n}) \rightarrow 0$.
\end{proof}

\begin{proof}[{\bf Proof of Corollary \ref{thm:prior_poly}}]
Since $s_i = (1+n\tau\gamma_i)^{-1}$ is a decreasing function in $\tau$, 
$$E(1-s_i \, | \, \bY) = \int_{\xi_n}^{\psi_n} E(1-s_i \, | \, \tau, \bY) \pi_n^{\tau}(\tau) d\tau \leq E(1-s_i \, | \, \tau = \psi_n, \bY).$$
Similarly, $$P(s_i < \eta \, |\, \bY) \leq P(s_i < \eta \, | \, \tau= \psi_n, \bY),$$ and $$P(s_i > \eta \, |\, \bY) \leq P(s_i > \eta \, | \, \tau= \xi_n, \bY).$$ The rest of the proof follows the proof of Theorem \ref{thm:tuning_poly}.
\end{proof}

\begin{proof}[{\bf Proof of Theorem \ref{thm:tuning_exp_vs}}] Similar to the proof of the variable selection consistency part of Theorem \ref{thm:tuning_poly}, the proof will be complete if we can show $\sum_{i\,\not\in\, \mathcal{A}} P\left(E(1 - s_i \, | \, \tau_n, \bY) > 1/2\right) \rightarrow 0$, as $n \rightarrow \infty$.
By Lemma \ref{lemma:noise_2}, 
$$P\left(E(1 - s_i|\bY) > \frac{1}{2}\right) \leq P\left(C\exp\left(\frac{n\hat{\beta}_i^2}{2\sigma^2}\right)n\tau_n\xi_n' > \frac{1}{2}\right) = P\left(n\hat\beta_i^2 > M_n'\right) = 2\left[1 - \Phi(\sqrt{M_n'})\right],$$ where $\xi_n'$, not depending on $i$, is a generic term that converges to 1 as $n\rightarrow \infty$ and $M_n' = -2\log(2Cn\tau_n\xi_n')$. If $n\tau_n \rightarrow 0$, then $M_n' \rightarrow \infty$. Hence,
\begin{align*}
\sum_{i\,\not\in\, \mathcal{A}} P\left(E(1 - s_i \, | \, \tau_n, \bY) > 1/2\right) \leq 2p_n\left[1 - \Phi(\sqrt{M_n'})\right] \sim \frac{2p_n\phi(\sqrt{M_n'})}{\sqrt{M_n'}} = \frac{2Cp_nn\tau_n\xi_n'}{\sqrt{-\pi\log(2Cn\tau_n\xi_n')}} \rightarrow 0,
\end{align*}
if $\frac{p_n n\tau_n}{\sqrt{\log(n\tau_n)}} \rightarrow 0$, as $n\rightarrow \infty$.
\end{proof}

\begin{proof}[{\bf Proof of Proposition \ref{prop:not_vs}}]
Notice that
\begin{align*}
E\left(1-s_{i} \, | \, \tau_n, \bY\right) &=  \frac{\int_{0}^{\infty}\frac{n\tau_{n}\gamma_{i}}{1+n\tau_{n}\gamma_{i}}\left(1+n\tau_{n}\gamma_{i}\right)^{-\frac{1}{2}}\exp\left\{ \frac{n\hat{\beta}_{i}^{2}}{2\sigma^{2}}\frac{n\tau_{n}\gamma_{i}}{1+n\tau_{n}\gamma_{i}}\right\} \pi\left(\gamma_{i}\right)d\gamma_{i}}{\int_{0}^{\infty}\left(1+n\tau_{n}\gamma_{i}\right)^{-\frac{1}{2}}\exp\left\{ \frac{n\hat{\beta}_{i}^{2}}{2\sigma^{2}}\frac{n\tau_{n}\gamma_{i}}{1+n\tau_{n}\gamma_{i}}\right\} \pi\left(\gamma_{i}\right)d\gamma_{i}} \\
& \geq \frac{\int_{0}^{\infty}\gamma_{i}\left(\frac{1}{n\tau_{n}}+\gamma_{i}\right)^{-\frac{3}{2}}\pi\left(\gamma_{i}\right)d\gamma_{i}}{\int_{0}^{\infty}\left(\frac{1}{n\tau_{n}}+\gamma_{i}\right)^{-\frac{1}{2}}\pi\left(\gamma_{i}\right)d\gamma_{i}}\exp\left\{ -\frac{n\hat{\beta}_{i}^{2}}{2\sigma^{2}}\right\}.
\end{align*}
Let $h_{n}{=}\frac{\int_{0}^{\infty}\gamma_{i}\left(\frac{1}{n\tau_{n}}+\gamma_{i}\right)^{-\frac{3}{2}}\pi\left(\gamma_{i}\right)d\gamma_{i}}{\int_{0}^{\infty}\left(\frac{1}{n\tau_{n}}+\gamma_{i}\right)^{-\frac{1}{2}}\pi\left(\gamma_{i}\right)d\gamma_{i}}$.
If $n\tau_n \rightarrow c \in (0, \infty]$ and $\int_0^{\infty} \gamma_i^{-\frac{1}{2}}\pi(\gamma_i)d\gamma_i < \infty$, by applying LDCTh to both the numerator and the denominator of $h_n$, we have $h_n$ converges to some positive constant that depends on $c$ and $\pi(\cdot)$ as $n\rightarrow \infty$.
Then, for any $i\notin \mathcal{A}$,
$$P(\mathcal{A}_n = \mathcal{A}) \leq P\left(E(1 - s_i \, | \, \tau_n, \bY) \leq 1/2\right)\leq P\left(h_{n}\exp\left\{ -\frac{n\hat{\beta}_{i}^{2}}{2\sigma^{2}}\right\} <\frac{1}{2}\right).$$
Note that $h_{n}\exp\left\{ -\frac{n\hat{\beta}_{i}^{2}}{2\sigma^{2}}\right\} $ converges in distribution to some distribution $Z$ with support on $(0,1)$,
so 
$$P\left(h_{n}\exp\left\{ -\frac{n\hat{\beta}_{i}^{2}}{2\sigma^{2}}\right\} <\frac{1}{2}\right)\rightarrow P\left(Z<\frac{1}{2}\right)<1.$$ Thus the HT procedure does not achieve variable selection consistency.
\end{proof}

\begin{proof}[{\bf Proof of Proposition \ref{prop:not_an}}]
Similar to the proof of Theorem \ref{thm:tuning_poly}, 
\begin{equation*}
\sqrt{n}\left(\hat{\beta}_{i}^{HT}-\beta_{i}^{0}\right) = \sqrt{n}\left(\hat{\beta}_{i}-\beta_{i}^{0}\right) - \sqrt{n}E(s_{i} \, | \, \tau_n, \bY)\hat{\beta}_{i} - \sqrt{n}\hat{\beta}_i^{PM}I(E(1-s_i \, | \, \tau_n, \bY) \leq 1/2).
\end{equation*}
For $i \in \mathcal{A}$, the third term in the right hand side converge to zero in probability. The first term has a normal distribution with mean 0 and variance $\sigma^2$.
The posterior density function of $s_{i}$ is
$$p\left(s_{i}\, | \, \tau_n, \bY\right)\propto s_{i}^{-\frac{3}{2}}\exp\left\{ -\frac{n\hat{\beta}_{i}^{2}}{2\sigma^{2}}s_{i}-\frac{b}{n\tau_{n}s_{i}}\right\}L\left(\frac{1}{n\tau_n}\left(\frac{1}{s_i}-1\right)\right) ,~0\leq s_{i}\leq1.$$
It is obvious that $\frac{m}{M}\tilde{S}^{(i)}_n \leq E(s_i \, | \, \tau_n, \bY) \leq \frac{M}{m}\tilde{S}^{(i)}_n$ almost surely, where $$\tilde{S}_n^{(i)}=\frac{\int_0^1 s_{i}^{-\frac{1}{2}}\exp\left\{ -\frac{n\hat{\beta}_{i}^{2}}{2\sigma^{2}}s_{i}-\frac{b}{n\tau_{n}s_{i}}\right\}ds_i}{\int_0^1 s_{i}^{-\frac{3}{2}}\exp\left\{ -\frac{n\hat{\beta}_{i}^{2}}{2\sigma^{2}}s_{i}-\frac{b}{n\tau_{n}s_{i}}\right\}ds_i}.$$
Notice that $$s_{i}^{-\frac{3}{2}}\exp\left\{ -\frac{n\hat{\beta}_{i}^{2}}{2\sigma^{2}}s_{i}-\frac{b}{n\tau_{n}s_{i}}\right\}I\left(0<s_{i}<\infty\right) = \left(\frac{\lambda_D}{2\pi}\right)^{-\frac{1}{2}}\exp\left(-\frac{\lambda_D}{\mu_D}\right)f(s_i;\lambda_D,\mu_D),$$ 
where $\lambda_D=2b/n\tau_{n}$, $\mu=\frac{\sqrt{2b}\sigma}{\abs{\hat{\beta}_{i}}}\left(n^{2}\tau_{n}\right)^{-\frac{1}{2}}$ and $f(x;\lambda, \mu) = \left(\frac{\lambda}{2\pi x^3}\right)^{\frac{1}{2}}\exp\left\{\frac{-\lambda\left(x-\mu\right)^2}{2\mu^2 x}\right\}I\left(0<x<\infty\right)$ is the probability density function of an Inverse Gaussian (IG) distribution with mean $\mu$ and shape parameter $\lambda$. 
According to \citet{shuster1968inverse}, the cdf of IG$(\mu,\lambda)$ can be expressed as $$F(x;\lambda,\mu) = \Phi\left(\sqrt{\frac{\lambda}{\mu}}\left(\frac{x}{\mu}-1\right)\right) + \exp\left(\frac{2\lambda}{\mu}\right)\Phi\left(-\sqrt{\frac{\lambda}{\mu}}\left(\frac{x}{\mu}+1\right)\right).$$ Thus the denominator of $\tilde{S}_n^{(i)}$ can be written as  $\left(\frac{\lambda_D}{2\pi}\right)^{-\frac{1}{2}}\exp\left(-\frac{\lambda_D}{\mu_D}\right)F(1;\lambda_D, \mu_D)$. With the transformation $s_i = 1/t_i$ and similar arguments for the denominator, the numerator of $\tilde{S}_n^{(i)}$ can be expressed as $\left(\frac{\lambda_N}{2\pi}\right)^{-\frac{1}{2}}\exp\left(-\frac{\lambda_N}{\mu_N}\right)F(1;\lambda_N, \mu_N)$, where $\lambda_N = n\hat\beta_i^2/\sigma^2$ and $\mu_N = \frac{n\sqrt{\tau_n}\abs{\hat\beta_i}}{\sqrt{2b\sigma}}$. By some simple calculations,
\begin{equation*}
\tilde{S}_n^{(i)} = \frac{\sqrt{2b\sigma}\left\{\Phi(b_n) - \exp(c_n)\Phi(-d_n)\right\}}{n\sqrt{\tau_n}\abs{\hat\beta_i}\left\{\Phi(b_n) + \exp(c_n)\Phi(-d_n)\right\}},
\end{equation*}
where $b_n = \sqrt{\frac{2b}{n\tau_n}}\left(\frac{n\sqrt{\tau_n}\abs{\hat\beta_i}}{\sqrt{2b}\sigma}-1\right)$, $c_n = \frac{2\sqrt{2b}\abs{\hat\beta_i}}{\sqrt{\tau_n}\sigma}$, and $d_n = \sqrt{\frac{2b}{n\tau_n}}\left(\frac{n\sqrt{\tau_n}\abs{\hat\beta_i}}{\sqrt{2b}\sigma}+1\right)$.

If $n^2\tau_n \rightarrow \infty$ as $n\rightarrow \infty$, $b_n \overset{p}{\rightarrow} +\infty$ and thus $\Phi(b_n) \overset{p}{\rightarrow} 1$. Combining  this with the fact that $\Phi(b_n) + \exp(c_n)\Phi(-d_n)$ is in $[0,1]$ since it equals $F(1, \lambda_D,\mu_D)$, we have $\exp(c_n)\Phi(-d_n) \overset{p}{\rightarrow} 0$. As a result, $\sqrt{n}\tilde{S}_n^{(i)} \overset{p}{\rightarrow} \infty$.

If $n^2\tau_n\rightarrow c \in (0,\infty)$ as $n \rightarrow \infty$, the limit of $b_n$ can be $+\infty$, $-\infty$, or some constant $r$. We will discuss the three cases separately. If $b_n \overset{p}{\rightarrow} +\infty$, by similar arguments as in the case $n^2\tau_n\rightarrow \infty$, we have $\sqrt{n}\tilde{S}_n^{(i)} \overset{p}{\rightarrow} \infty$. If $b_n \overset{p}{\rightarrow} -\infty$,
\begin{equation*}
\frac{\exp(c_n)\Phi(-d_n)}{\Phi(b_n)} \sim \frac{\exp(c_n)\phi(d_n)/d_n}{\phi(-b_n)/(-b_n)} = \frac{\sqrt{2b}\sigma-n\sqrt{\tau_n}\abs{\hat\beta_i}}{\sqrt{2b}\sigma+n\sqrt{\tau_n}\abs{\hat\beta_i}} \overset{p}{\rightarrow} \frac{\sqrt{2b}\sigma-\sqrt{c}\abs{\beta_i^0}}{\sqrt{2b}\sigma+\sqrt{c}\abs{\beta_i^0}}. 
\end{equation*}
Therefore, $\tilde{S}_n^{(i)} \overset{p}{\rightarrow} 1$ and $\sqrt{n}\tilde{S}_n^{(i)} \overset{p}{\rightarrow} \infty$. If $b_n \overset{p}{\rightarrow} r$, since $c_n - \frac{1}{2}d_n^2 + \frac{1}{2}b_n^2 = 0$, we have $c_n - \frac{1}{2}d_n^2 \overset{p}{\rightarrow} -\frac{1}{2}r^2$. Therefore, $\exp(c_n)\Phi(-d_n) \sim \exp(c_n)\phi(d_n)/d_n \overset{p}{\rightarrow} 0$ and $\sqrt{n}\tilde{S}_n^{(i)} \overset{p}{\rightarrow} \infty.$

If $n^2\tau_n \rightarrow 0$, $b_n \overset{p}{\rightarrow} -\infty$. By the famous inequality $\frac{x^2}{1+x^2} \leq \frac{x(1-\Phi(x))}{\phi(x)} \leq 1$,
\begin{equation*}
\tilde{S}_n^{(i)} \geq \frac{\sqrt{2b}\sigma}{n\sqrt{\tau_n}|\hat\beta_i|} \frac{\frac{b_n^2}{1+b_n^2}\frac{\phi(b_n)}{b_n} + \exp(c_n)\frac{\phi(d_n)}{d_n}}{\frac{\phi(b_n)}{b_n} - \exp(c_n)\frac{\phi(d_n)}{d_n}}\\
=1 - \frac{1}{2}\frac{1+\frac{\sqrt{2b}\sigma}{n\sqrt{\tau_n}|\hat\beta_i|}}{ 1 + \frac{2b}{n\tau_n}\left(\frac{n\sqrt{\tau_n}|\hat\beta_i|}{\sqrt{2b}\sigma}-1\right)^2} \overset{p}{\rightarrow} 1.
\end{equation*}
In all the cases, $\sqrt{n}\tilde{S}_n^{(i)} \overset{p}{\rightarrow} \infty$, as $n \rightarrow \infty$. Thus $\sqrt{n} E(s_i \, | \, \tau_n, \bY) \overset{p}{\rightarrow} \infty$ and $\sqrt{n}(\hat\beta_i^{PM} - \hat\beta_i^0) \not\rightarrow \mbox{N}(0, \sigma^2)$.
\end{proof}

\begin{proof}[{\bf Proof of Theorem \ref{thm:tuning_exp_an}}]
Following the proof of Theorem \ref{prop:not_an}, 
$n\sqrt{\tau_n}\tilde{S}_n^{(i)} \overset{p}{\rightarrow} \sqrt{2b}\sigma/|\beta_i^0|$. This completes the proof.
\end{proof}

\begin{proof}[{\bf Proof of Theorem \ref{thm:tuning_exp}}]
Theorem \ref{thm:tuning_exp} is a direct result of Propositions \ref{prop:not_vs} and \ref{prop:not_an}.
\end{proof}

\bibliographystyle{apalike}
\bibliography{vs}

\end{document}